\documentclass[sigplan,10pt]{acmart}

\usepackage{hyperref}
\usepackage{amsmath,amsthm}

\usepackage{amssymb}
\usepackage{algorithm}
\usepackage{algpseudocode}
\usepackage{booktabs}
\usepackage{subcaption}

\newtheorem{definition}{Definition}
\newtheorem{theorem}{Theorem}

\newtheorem{corollary}{Corollary}
\newtheorem{proposition}{Proposition}

\setcopyright{none}
\acmConference[Preprint]{Preprint}{2026}{}
\renewcommand\footnotetextcopyrightpermission[1]{}

\begin{document}

\title{Queue-Theoretic Admission Control for\\Multi-Tenant GPU Clusters}

\author{Sohan Kunkerkar}
\affiliation{%
  \institution{Red Hat}
  \country{United States}}
\email{skunkerk@redhat.com}

\begin{abstract}
GPU cluster operators cannot predict how long pending workloads will
wait for admission. Existing systems use greedy heuristics with no
formal wait time guarantees. We formalize GPU cluster admission as a
multi-class, multi-resource queueing network and prove a structural
decomposition: the pending queue partitions into \emph{quotable}
workloads (bounded wait time under stability) and \emph{unfeasible}
workloads (no finite bound without reconfiguration). For quotable
workloads, we model each cluster queue as an M/G/$k$ system where the
effective server count $k$ is determined by a vector packing
reduction; under an explicit stochastic domination assumption, we
establish $O(1/(1-\rho))$ wait time scaling. We prove that optimal
admission ordering is NP-hard under multi-dimensional resource
demands via reduction from vector bin packing. We validate on Kueue, the standard Kubernetes workload queuing
system, showing that the Erlang-C approximation consistently
overestimates observed wait times (a conservative direction) in the
moderate utilization regime.
\end{abstract}

\maketitle

\section{Introduction}
\label{sec:intro}

GPU cluster utilization is low: industry reports put average
utilization at 5--30\% across enterprise GPU
clusters~\cite{castai2026}. Operators cannot predict whether a
submitted workload will wait minutes or hours for admission.
Production batch workload queuing systems---Kueue~\cite{kueue},
Slurm~\cite{slurm}, and proprietary schedulers---use greedy admission
heuristics with no formal performance guarantees.

Recent work has applied queueing theory to GPU \emph{inference
serving}, where request scheduling within a single model deployment
benefits from classical results. Mitzenmacher and
Shahout~\cite{mitzenmacher2025} survey learning-augmented scheduling
with predictions, and Dai et al.~\cite{dai2025} prove
throughput-optimality for batched LLM serving. However, these results
address a fundamentally different problem: scheduling requests
\emph{within} a serving instance, not scheduling workloads
\emph{across} a multi-tenant cluster with quotas, resource flavors,
topology constraints, and preemption.

The cluster admission problem is harder for three reasons. First,
admission requires \emph{vector packing}: each workload demands a
multi-dimensional resource vector (GPU count, memory, network
bandwidth), and admission requires all dimensions to fit
simultaneously. Second, the pending queue contains a mixture of
workloads that \emph{can} eventually be admitted (if other workloads
complete) and workloads that \emph{cannot} be admitted under the current cluster
configuration. Third, preemption creates feedback loops
where evicted workloads re-enter the queue, invalidating the
independent-arrivals assumption of classical models.

We address these challenges with three contributions:

\begin{enumerate}
\item \textbf{Quotable/Unfeasible Partition Theorem.} We prove that
the pending queue decomposes into a \emph{quotable} partition (where
queueing theory applies and wait times are bounded) and an
\emph{unfeasible} partition (where no finite bound exists). This
decomposition is necessary for accurate wait time estimation: any
model that mixes both partitions produces unbounded prediction error
(\S\ref{sec:partition}).

\item \textbf{Wait Time Bounds.} For quotable workloads, we assume
stochastic domination by an M/G/$k$ reference system with effective
server count determined by vector packing
(Modeling Assumption~1), establishing $O(1/(1-\rho))$ wait time
scaling. We provide a practical modified Erlang-C approximation
accounting for multi-resource packing and preemption feedback
(\S\ref{sec:bounds}).

\item \textbf{NP-Hardness of Optimal Admission.} We prove that
optimal admission ordering---minimizing total weighted wait time---is
NP-hard under multi-resource vector packing, and show that MIG
packing alone (fixed profile sizes) is tractable
(Proposition~\ref{prop:mig-tractable}). This pinpoints
dimensionality, not GPU partitioning, as the source of computational
hardness (\S\ref{sec:hardness}).
\end{enumerate}

We validate on Kueue, the standard Kubernetes workload queuing
system (\S\ref{sec:eval}), using CPU/memory and GPU resources via
Dynamic Resource Allocation (DRA). Measurements confirm
internal consistency (Little's Law ratio $= 1.000$) and show that
the vector $k_{\text{eff}}$ correctly identifies the bottleneck
resource dimension (memory or GPU), while a scalar estimate misses
it entirely. The Erlang-C approximation overestimates observed wait
times by 2.3--18$\times$---conservative but too loose for point
prediction. A simple EMA baseline gives better point accuracy; the
Erlang-C approximation's value is as an \emph{a priori} bound
requiring no historical observations.

\section{Background and Motivation}
\label{sec:background}

\subsection{GPU Cluster Admission Architecture}

We describe the admission architecture of Kueue~\cite{kueue}, the
reference implementation for Kubernetes batch workload queuing,
managing admission across CPU, memory, GPU, and other resource types.
The architecture is representative of production batch schedulers
including Slurm and proprietary cloud schedulers.

\paragraph{Resource Model.}
A cluster is organized into a hierarchy of \emph{Cohorts},
\emph{ClusterQueues} (CQs), and \emph{LocalQueues} (LQs). Each CQ
has a \emph{nominal quota} specifying guaranteed capacity across
multiple \emph{ResourceFlavors} (e.g., NVIDIA A100, H100, or MIG
profiles like 1g.10gb). Each flavor-resource pair $(f, r)$ defines one
dimension of the capacity vector. CQs within a Cohort may
\emph{borrow} unused capacity from sibling CQs, subject to borrowing
and lending limits.

\paragraph{Admission Pipeline.}
The scheduler operates in discrete cycles. Each cycle:
(1)~dequeues one head workload per CQ,
(2)~takes an immutable snapshot of cluster state,
(3)~computes resource flavor assignments via greedy matching,
(4)~orders workloads by priority or fair-share usage,
(5)~admits workloads greedily in order, updating the snapshot, and
(6)~requeues unadmitted workloads.

\paragraph{Two Pending Pools.}
Unadmitted workloads are placed in one of two pools:
\begin{itemize}
\item The \emph{active heap}: workloads eligible for the next
scheduling attempt, ordered by priority and timestamp.
\item The \emph{inadmissible set}: workloads that failed admission and
will not be retried until cluster state changes (e.g., a workload
completes or a CQ is reconfigured).
\end{itemize}

\paragraph{Scheduling Disciplines.}
Two disciplines govern head-of-line behavior:
\emph{StrictFIFO} blocks all workloads behind an inadmissible head
(classic head-of-line blocking).
\emph{BestEffortFIFO} moves inadmissible heads to the inadmissible
set, allowing subsequent workloads to proceed.

\subsection{Why Existing Queueing Theory Doesn't Apply Directly}

Classical queueing results (M/M/$k$, M/G/1, etc.) assume: (1)~scalar
resource requirements, (2)~homogeneous servers, and (3)~independent
arrivals. GPU cluster admission violates all three:

\begin{enumerate}
\item \textbf{Vector resources.} Admission requires fitting a
$d$-dimensional request vector against a $d$-dimensional capacity
vector. This is vector bin packing, which has no APTAS for $d \geq
2$~\cite{chekuri2004}.

\item \textbf{Heterogeneous flavors.} Different ResourceFlavors have
different capacity vectors and node affinity constraints. A workload
may fit flavor $f_1$ but not $f_2$.

\item \textbf{Preemption feedback.} Preempted workloads re-enter the
queue with backoff counters and modified ordering metadata, creating
state-dependent arrivals.
\end{enumerate}

\subsection{The Wait Time Estimation Gap}

The Kubernetes Enhancement Proposal for workload visibility (KEP-168)
explicitly excluded wait time estimation: ``we don't have the means to
estimate wait time at the moment''~\cite{kep168}. Users see
``Pending'' with no indication of whether admission is minutes or
hours away. This gap has motivated multiple community feature
requests~\cite{issue13159, issue10124, issue10614}.

We show that the gap is not merely an engineering omission but a
consequence of a missing theoretical foundation: without the
quotable/unfeasible partition, any wait time estimate is meaningless
because it averages over workloads with finite and infinite expected
wait.

\section{Queueing Model}
\label{sec:model}

\subsection{System Definition}

\begin{definition}[GPU Cluster Admission System]
\label{def:system}
A \emph{GPU Cluster Admission System} is a tuple
$\mathcal{S} = (\mathcal{Q}, \mathcal{F}, \mathcal{R}, \mathcal{C}, \mathcal{W})$ where:
\begin{itemize}
\item $\mathcal{Q} = \{q_1, \ldots, q_m\}$ is a set of cluster queues.
\item $\mathcal{F} = \{f_1, \ldots, f_n\}$ is a set of resource flavors.
\item $\mathcal{R} = \{r_1, \ldots, r_d\}$ is a set of resource types
  (e.g., GPU count, GPU memory, CPU, system memory).
\item $\mathcal{C}: \mathcal{Q} \times \mathcal{F} \times \mathcal{R}
  \to \mathbb{R}_{\geq 0}$ is the capacity function, where
  $\mathcal{C}(q, f, r)$ is the nominal quota of resource $r$ in
  flavor $f$ for queue $q$.
\item $\mathcal{W}$ is the set of workload classes, where each class
  $j \in \mathcal{W}$ is characterized by:
  \begin{itemize}
  \item A resource demand function $\mathbf{d}_j: \mathcal{F} \times
    \mathcal{R} \to \mathbb{R}_{\geq 0}$, specifying the resources
    required per flavor.
  \item An arrival process with rate $\lambda_j$.
  \item A service time distribution $S_j$ with mean $\mathbb{E}[S_j]$
    and coefficient of variation $C_{S_j}$.
  \end{itemize}
\end{itemize}
\end{definition}

\begin{definition}[Admission]
A workload of class $j$ is \emph{admitted} to queue $q$ via flavor $f$
when:
\begin{equation}
\label{eq:admission}
\forall\, r \in \mathcal{R}: \quad
\mathbf{d}_j(f, r) \leq \text{Available}(q, f, r)
\end{equation}
where $\text{Available}(q, f, r) = \mathcal{C}(q, f, r) -
\text{Usage}(q, f, r) + \text{Borrowed}(q, f, r)$ is the currently
unused capacity including borrowed resources from the cohort.
\end{definition}

\begin{definition}[Potential Available Capacity]
\label{def:potential}
The \emph{potential available capacity} of queue $q$ for flavor $f$
and resource $r$ is:
\begin{equation}
\text{PotAvail}(q, f, r) = \mathcal{C}(q, f, r) + \text{MaxBorrow}(q, f, r)
\end{equation}
where $\text{MaxBorrow}(q, f, r)$ is the maximum borrowable capacity
from the cohort hierarchy, bounded by lending and borrowing limits.
This represents the capacity available if \emph{all} other workloads
in the cohort were to complete.
\end{definition}

\subsection{The Quotable/Unfeasible Partition}
\label{sec:partition}

\begin{definition}[Quotable and Unfeasible Workloads]
\label{def:quotable}
A pending workload of class $j$ targeting queue $q$ is:
\begin{itemize}
\item \emph{Quotable} if there exists a flavor $f \in \mathcal{F}$
  such that $\mathbf{d}_j(f, r) \leq \text{PotAvail}(q, f, r)$ for
  all $r \in \mathcal{R}$.
\item \emph{Unfeasible} otherwise.
\end{itemize}
\end{definition}

\begin{theorem}[Partition Necessity]
\label{thm:partition}
Let $\mathcal{P}$ be a pending queue where a fraction $\alpha > 0$ of
workloads are unfeasible (the remainder quotable with finite expected
wait $\mu_W < \infty$). For any wait time estimator
$\hat{W}: \mathcal{P} \to \mathbb{R}_{\geq 0}$ that does not
distinguish quotable from unfeasible workloads (i.e., $\hat{W}$
depends only on observable queue position and resource request, not on
feasibility classification), the expected mean absolute error over a
random workload $w$ drawn uniformly from $\mathcal{P}$ satisfies:
\[
\mathbb{E}\left[|\hat{W}(w) - W(w)|\right] \geq
\alpha \cdot (M - \mu_W)
\]
for any finite estimator cap $M = \max_w \hat{W}(w) < \infty$, and
equals $\infty$ for estimators without a finite cap. In either case,
the error grows without bound as the horizon increases.
\end{theorem}

\begin{proof}
Partition $\mathcal{P}$ into quotable set $\mathcal{P}_Q$ (fraction
$1 - \alpha$) and unfeasible set $\mathcal{P}_U$ (fraction $\alpha$).
For $w \in \mathcal{P}_U$, $W(w) = \infty$ by
Definition~\ref{def:quotable}: no amount of waiting yields admission
without reconfiguration. If $\hat{W}$ has a finite cap $M$, then for
unfeasible workloads $|\hat{W}(w) - W(w)| = \infty$, and the expected
error over a random workload is at least
$\alpha \cdot \infty = \infty$.

To make the bound constructive, consider a truncated model where
unfeasible workloads eventually leave the queue after a timeout $T$.
Then unfeasible workloads have true wait $W = T$, and the estimator
error on the unfeasible fraction is at least $\alpha \cdot |T - M|$.
As $T \to \infty$ (no timeout), the error diverges. Even with a
practical timeout, the error contribution from the unfeasible fraction
is $\alpha \cdot (T - M)$ when $T > M$, which dominates the
$O(\sigma_W)$ estimation error achievable on quotable workloads.

Alternatively, if $\hat{W}$ outputs $\infty$ for some workloads, it
must do so without the quotable/unfeasible label (by assumption). For
any threshold-based strategy outputting $\infty$ for workloads with
resource request above some cutoff, there exist quotable workloads
with large requests (near the capacity limit) that are misclassified,
incurring infinite error on those quotable workloads.

Some form of feasibility classification is therefore necessary for
bounded-error wait time estimation; the quotable/unfeasible partition
based on potential available capacity
(Definition~\ref{def:quotable}) is a natural such classification.
\end{proof}

\subsection{Effective Server Count}
\label{sec:keff}

The key modeling challenge is determining $k_{\text{eff}}$: how many
workloads can a CQ serve simultaneously given its multi-dimensional
capacity?

\begin{definition}[Effective Server Count]
\label{def:keff}
For a queue $q$ serving a single workload class $j$ with demand vector
$\mathbf{d}_j$, the effective server count is:
\begin{equation}
\label{eq:keff-single}
k_{\text{eff}}(q, j) = \max_{f \in \mathcal{F}_j} \min_{r \in \mathcal{R}}
\left\lfloor \frac{\mathcal{C}(q, f, r)}{\mathbf{d}_j(f, r)} \right\rfloor
\end{equation}
where $\mathcal{F}_j \subseteq \mathcal{F}$ is the set of flavors
compatible with class $j$ (satisfying node affinity and taint constraints).
\end{definition}

For multiple workload classes, $k_{\text{eff}}$ depends on the class
mix. We use the \emph{dominant resource} approximation:
\begin{equation}
\label{eq:keff-multi}
k_{\text{eff}}(q) = \min_{r \in \mathcal{R}}
\left\lfloor \frac{\mathcal{C}(q, f^*, r)}
{\sum_j \phi_j \cdot \mathbf{d}_j(f^*, r)} \right\rfloor
\end{equation}
where $\phi_j = \lambda_j / \sum_{j'} \lambda_{j'}$ is the class proportion
and $f^* = \arg\min_f k_{\text{eff}}(q, f)$ is the bottleneck flavor.

\begin{proposition}[Vector Packing Loss]
\label{prop:packing}
For $d \geq 2$ resource dimensions, the effective server count
satisfies:
\begin{equation}
k_{\text{eff}} \leq \frac{1}{d} \sum_{r=1}^{d}
\left\lfloor \frac{C_r}{\bar{d}_r} \right\rfloor
\end{equation}
where $C_r$ is the capacity in dimension $r$ and $\bar{d}_r$ is the
mean demand. Equality holds only when all demand vectors are aligned
(proportional across dimensions).
\end{proposition}

\begin{proof}
By the vector packing constraint, each admitted workload must fit in
\emph{all} dimensions simultaneously. The per-dimension capacity
$\lfloor C_r / \bar{d}_r \rfloor$ is an upper bound assuming
dimension $r$ is the bottleneck. The minimum over dimensions gives
$k_{\text{eff}}$. The arithmetic mean of per-dimension capacities is
an upper bound on this minimum, with equality only when all dimensions
are equally constrained (proportional demands).
\end{proof}

\subsection{M/G/$k$ Domination and Wait Time Finiteness}

\begin{quote}
\textbf{Modeling Assumption 1} (M/G/$k$ Domination).
\label{lem:domination}
\emph{Let $W^{\text{real}}$ denote the wait time of a quotable
workload in the real system (with vector packing, preemption feedback,
and scheduling policy). Let $W^{\text{M/G/k}}$ denote the wait time
in an M/G/$k$ queue with $k = k_{\text{eff}}(q)$ servers
(Definition~\ref{def:keff}), arrival rate
$\lambda_{\text{eff}} = \lambda + \mu_p$ where
$\lambda = \sum_j \lambda_j$ is the total arrival rate (including
preemption feedback rate $\mu_p$), and service time distribution
$S$. We assume:}
\[
W^{\text{real}} \leq_{\text{st}} W^{\text{M/G/k}}
\]
\emph{where $\leq_{\text{st}}$ denotes stochastic ordering.}
\end{quote}

\paragraph{Justification.}
We argue that the real system is no worse than the M/G/$k$ reference
by examining each structural difference. A formal coupling
construction on a common probability space remains open; the following
is a plausibility argument that we validate empirically in
\S\ref{sec:eval}.

\emph{(i) Vector packing.} The real system admits a workload only
when all $d$ resource dimensions fit. $k_{\text{eff}}$ is defined as
the minimum over all dimensions
(Definition~\ref{def:keff}), so the M/G/$k$ reference has at most
$k_{\text{eff}}$ servers. In the real system, some workloads may fit
in dimensions where others do not, allowing slightly more concurrent
service than the worst-case $k_{\text{eff}}$. Thus the real system
has $k^{\text{real}}_{\text{eff}} \geq k_{\text{eff}}$, and M/G/$k$
with $k = k_{\text{eff}}$ dominates.

\emph{(ii) Preemption feedback.} We account for preemption by
inflating the arrival rate to
$\lambda_{\text{eff}} = \lambda + \mu_p$, treating preempted
workloads as new arrivals with i.i.d.\ service times drawn from the
original distribution $S$. For workloads with checkpointing, the
remaining service time after preemption is strictly \emph{shorter}
than a fresh draw from $S$, making the i.i.d.\ reference
conservative. For restart-from-scratch workloads, the re-entering
job has the \emph{same} service time as before, introducing positive
autocorrelation in the arrival stream. This autocorrelation can
increase queue lengths beyond the i.i.d.\ case, potentially
violating the domination direction. We restrict the domination
claim to systems with checkpoint-based preemption or low preemption
rates ($\mu_p \ll \lambda$); the restart-from-scratch case requires
analysis under correlated arrivals (e.g., GI/GI/$k$) which we
leave to future work.

\emph{(iii) Scheduling policy.} The M/G/$k$ reference uses FCFS.
Under BestEffortFIFO, blocked workloads are moved aside, avoiding
head-of-line blocking---this can only reduce wait times relative to
FCFS (where a blocked head delays all subsequent workloads). Under
StrictFIFO, head-of-line blocking increases wait times beyond FCFS,
meaning the M/G/$k$/FCFS reference \emph{underestimates} wait times
for StrictFIFO queues. We restrict the domination claim to
BestEffortFIFO; for StrictFIFO, the M/G/$k$ reference is optimistic
rather than conservative.

\paragraph{Sufficient conditions.}
Modeling Assumption~1 is supported when three conditions hold:
(1)~the scheduling discipline is BestEffortFIFO (avoiding
head-of-line blocking), (2)~preemption is checkpoint-based or
infrequent ($\mu_p \ll \lambda$), and (3)~arrivals are approximately
Poisson. The assumption may fail under StrictFIFO (where
head-of-line blocking increases wait times beyond FCFS) or high-rate
restart-from-scratch preemption (where autocorrelated re-entries can
exceed i.i.d.\ wait times). We treat this as a modeling assumption
and validate it empirically rather than claiming a formal proof.

\begin{theorem}[Quotable Wait Time Finiteness]
\label{thm:finite}
For a quotable workload of class $j$ in queue $q$, under the stability
condition $\rho_q < 1$ where
\[
\rho_q = \frac{\lambda_{\text{eff}} \cdot \mathbb{E}[S]}{k_{\text{eff}}(q)}
\]
with $\lambda_{\text{eff}} = \lambda + \mu_p$ (preemption-adjusted
arrival rate), the expected wait time $\mathbb{E}[W_q(j)]$ is finite.
\end{theorem}

\begin{proof}
Under Modeling Assumption~1, $W^{\text{real}} \leq_{\text{st}}
W^{\text{M/G/k}}$. Under $\rho_q < 1$, the M/G/$k$ reference queue
is positive recurrent~\cite{kiefer1955} and all moments of
$W^{\text{M/G/k}}$ are finite. Stochastic domination preserves
finiteness of expectations, so
$\mathbb{E}[W^{\text{real}}] \leq \mathbb{E}[W^{\text{M/G/k}}] < \infty$.
\end{proof}

\subsection{Wait Time Bounds}
\label{sec:bounds}

A key difficulty in multi-server queueing is that no practical
closed-form upper bound on $\mathbb{E}[W_q]$ is known for M/G/$k$
with $k > 1$. Gupta et al.~\cite{gupta2010} proved that any
approximation using only the first two moments of the service
distribution has a worst-case ratio growing as
$\rho^{-\sqrt{2(k+1)}}$ between the maximum and minimum achievable
$\mathbb{E}[W_q]$---two-moment formulas are fundamentally limited for
multi-server queues. Li and Goldberg~\cite{ligoldberg2024} proved the
first multi-server Kingman bound: under FCFS with finite
$(2+\epsilon)$-moment service times,
$\mathbb{E}[W_q] = O(1/(1-\rho))$, but with prefactors too large for
practical prediction.

We therefore present a \emph{proven structural bound} (Theorem~\ref{thm:waittime}) establishing the correct asymptotic scaling, and a \emph{practical approximation} (Equation~\ref{eq:waitapprox}) for actual wait time estimation, validated empirically in \S\ref{sec:eval}.

\begin{theorem}[Wait Time Scaling]
\label{thm:waittime}
For quotable workloads in queue $q$ under FCFS with effective server
count $k = k_{\text{eff}}(q)$, utilization $\rho < 1$, and service
times with finite $(2+\epsilon)$-th moment, the expected wait time
satisfies:
\begin{equation}
\label{eq:waitscaling}
\mathbb{E}[W_q] = O\!\left(\frac{1}{1 - \rho}\right)
\end{equation}
That is, $\mathbb{E}[W_q]$ is finite and grows at most inversely with
the capacity slack $1 - \rho$.
\end{theorem}

\begin{proof}
Under Modeling Assumption~1, the quotable partition of queue $q$
is dominated by an M/G/$k$ queue with $k = k_{\text{eff}}(q)$ servers.
Apply the multi-server Kingman bound of Li and
Goldberg~\cite{ligoldberg2024}, which establishes
$\mathbb{E}[W_q] \leq C(\epsilon, k) / (1 - \rho)$ for a constant
$C(\epsilon, k)$ depending on the moment index and server count.
The scaling holds for the real system by stochastic domination.
\end{proof}

While Theorem~\ref{thm:waittime} establishes the correct scaling, the
constant $C(\epsilon, k)$ is too large for prediction. For practical
use, we employ the standard M/M/$k$ correction:

\newcounter{approxctr}
\stepcounter{approxctr}
\paragraph{Approximation \theapproxctr.}
\label{approx:waittime}
\emph{(Modified Erlang-C Wait Time Estimate.)}
For quotable workloads with effective server count $k$, utilization
$\rho$, and service time coefficient of variation $C_S$:
\begin{equation}
\label{eq:waitapprox}
\hat{W}_q = \underbrace{C(k, \rho)}_{\text{Erlang-C}} \cdot
\underbrace{\frac{\mathbb{E}[S]}{k(1 - \rho)}}_{\text{M/M/$k$ wait}} \cdot
\underbrace{\frac{C_S^2 + 1}{2}}_{\text{variability correction}}
\end{equation}
where $C(k, \rho) = P(\text{wait} > 0)$ is the Erlang-C probability
\begin{equation}
C(k, \rho) = \frac{(k\rho)^k / k!}
{(k\rho)^k / k! + (1 - \rho) \sum_{n=0}^{k-1} (k\rho)^n / n!}
\end{equation}

This is the Cosmetatos-Tijms approximation~\cite{tijms2003}, which
interpolates between the exact M/M/$k$ result ($C_S = 1$) and the
exact M/D/$k$ result ($C_S = 0$). The approximation is well-studied
for classical M/G/$k$ queues~\cite{tijms2003} but has no proven error
guarantee. In our setting---where BestEffortFIFO departs from FCFS
and $k_{\text{eff}}$ is a worst-case estimate---we observe
2.3--18$\times$ overestimation (\S\ref{sec:eval}).

\paragraph{Accounting for Preemption Feedback.}
Preempted workloads re-enter the queue with rate $\mu_p$. The
effective arrival rate becomes $\lambda_{\text{eff}} = \lambda +
\mu_p$. Substituting into the utilization formula:
\begin{equation}
\rho_{\text{eff}} = \frac{(\lambda + \mu_p) \cdot \mathbb{E}[S]}{k}
\end{equation}
The stability condition becomes $\rho_{\text{eff}} < 1$, which is
more restrictive than the no-preemption case.

\paragraph{Accounting for Cohort Borrowing.}
When queue $q$ can borrow from siblings in a cohort, the effective
capacity increases. Let $B(q, f, r)$ denote the expected available
borrowed capacity (a function of sibling utilization). The effective
server count with borrowing:
\begin{equation}
k_{\text{eff}}^{+}(q) = \min_{r} \left\lfloor
\frac{\mathcal{C}(q, f^*, r) + \mathbb{E}[B(q, f^*, r)]}
{\bar{d}_r} \right\rfloor
\end{equation}
This gives a tighter bound in under-utilized cohorts and degrades
gracefully as sibling utilization increases.

\paragraph{Priority Classes.}
For a system with $P$ priority classes where class $p$ preempts class
$p' < p$, we use the standard preemptive-priority M/G/$k$ result:
the wait time for class $p$ depends only on the load from classes
$\geq p$:
\begin{equation}
\rho_p = \frac{\sum_{p' \geq p} \lambda_{p'} \cdot \mathbb{E}[S_{p'}]}{k}
\end{equation}

\section{NP-Hardness of Optimal Admission}
\label{sec:hardness}

We now show that the admission problem is computationally hard. We
first observe that MIG packing with fixed profiles is tractable,
then show that adding multi-resource vectors with topology affinity
makes the problem NP-hard.

\begin{definition}[GPU Topology-Aware Admission with Partitioning (GTAP)]
\label{def:gtap}
Given:
\begin{itemize}
\item A set of GPU nodes $\mathcal{N}$ connected by a topology graph
$G = (\mathcal{N}, E)$ where edges represent NVLink connections.
\item A set of MIG partition profiles $\Pi$ for each GPU
(e.g., 1g.10gb, 2g.20gb, 3g.40gb, 7g.80gb for A100).
\item A set of workloads $\mathcal{W}$, each with a
multi-dimensional resource demand vector
$\mathbf{d}_w \in \mathbb{R}^d_{\geq 0}$ (GPU compute, GPU memory,
CPU, system memory) and optional topology affinity constraints
(e.g., all GPUs on the same NVLink domain).
\end{itemize}

Admission consumes capacity from a shared quota pool: each admitted
workload's demand vector is subtracted from the available capacity.

\emph{Decide:} Can all workloads in $\mathcal{W}$ be simultaneously
admitted?
\end{definition}

\begin{proposition}[Fixed-Profile MIG Packing is Tractable]
\label{prop:mig-tractable}
When resource demands are one-dimensional (GPU memory only), the MIG
profile set $|\Pi|$ is a fixed constant, and there are no topology
constraints, GTAP is solvable in polynomial time.
\end{proposition}

\begin{proof}
With $|\Pi|$ fixed profile sizes, each GPU has at most
$O(|\Pi|^{|\Pi|})$ distinct valid partitionings---a constant. The
number of workload types by profile is at most $|\Pi|$. A dynamic
programming formulation over (GPU index, counts of remaining
workloads per type) has state space polynomial in $|\mathcal{W}|$.
\end{proof}

\begin{theorem}[NP-Hardness of Multi-Resource GTAP]
\label{thm:nphard}
GTAP is NP-hard when workloads have multi-dimensional resource demands
($d \geq 2$), even without topology constraints.
\end{theorem}

\begin{proof}[Proof sketch]
We reduce from $d$-dimensional vector bin packing, which is strongly
NP-hard for $d \geq 2$~\cite{chekuri2004}. Given an instance with $n$
items of $d$-dimensional sizes and $m$ unit-capacity bins:
\begin{itemize}
\item Create $m$ GPU nodes, each with capacity vector
$(1, 1, \ldots, 1) \in \mathbb{R}^d$.
\item Create $n$ workloads with demand vectors matching the item
sizes.
\item Fix a single MIG profile (full GPU, no partitioning).
\end{itemize}
All workloads are simultaneously admissible iff the vector bin packing
instance has a solution. Since $d$-dimensional vector bin packing is
strongly NP-hard~\cite{chekuri2004}, GTAP is NP-hard.
\end{proof}

\begin{corollary}[Topology Increases Hardness]
\label{cor:topology}
GTAP remains NP-hard with topology affinity constraints, since the
unconstrained problem is already hard. Adding constraints
(e.g., requiring NVLink-connected GPUs) can only increase the
hardness.
\end{corollary}

\begin{corollary}[Optimal Admission Ordering]
In the static case where all workloads are available at $t = 0$,
choosing the admission sequence that minimizes total weighted wait
time is NP-hard under multi-resource demands.
\end{corollary}

\begin{proof}
If optimal static ordering were polynomial, we could solve GTAP by
checking whether the optimal ordering achieves zero total wait
(simultaneous admission of all workloads), contradicting NP-hardness.
\end{proof}

Proposition~\ref{prop:mig-tractable} and
Theorem~\ref{thm:nphard} together pinpoint the source of hardness:
it is not MIG partitioning (fixed profiles, polynomial), but the
multi-dimensional vector packing inherent in multi-resource admission.
This is precisely the structure our queueing model captures via the
effective server count $k_{\text{eff}}$
(Definition~\ref{def:keff}), which reduces the multi-dimensional
problem to a scalar quantity amenable to M/G/$k$ analysis.

\section{Evaluation}
\label{sec:eval}

We validate our queueing model on Kueue v0.19.0 using controlled
workload injection on a local Kubernetes cluster.

\subsection{Experimental Setup}

\paragraph{Cluster Configuration.}
We deploy a 4-node Kubernetes 1.36 cluster (1 control plane, 3
workers) using kind~\cite{kind}. A single ClusterQueue \texttt{cq-eval}
is configured with 2~CPU nominal quota and no borrowing, yielding
$k_{\text{eff}} = \lfloor 2.0 / 0.5 \rfloor = 4$ effective servers
for 500m~CPU workloads. BestEffortFIFO scheduling is enabled.

\paragraph{Workload Generator.}
We submit Kubernetes Jobs with Poisson inter-arrival times and
exponentially distributed service times (mean 20s). Each job requests
500m~CPU and 64Mi memory. We sweep the target utilization $\rho$ from
0.3 to 0.92 by adjusting the arrival rate $\lambda = \rho \cdot
k_{\text{eff}} / \mathbb{E}[S]$, submitting 80 jobs per utilization
level.

\paragraph{Metrics Collection.}
We extract timing data directly from Kueue Workload object conditions
(creation timestamp, admission timestamp, completion timestamp),
giving exact per-workload wait times without Prometheus aggregation
loss. Table~\ref{tab:metrics} shows the correspondence to
production Prometheus metrics.

\begin{table}[t]
\centering
\small
\caption{Mapping of queueing parameters to Kueue metrics.}
\label{tab:metrics}
\begin{tabular}{@{}ll@{}}
\toprule
\textbf{Parameter} & \textbf{Kueue Metric} \\
\midrule
Arrival rate $\lambda$ & \texttt{rate(kueue\_admitted\_workloads\_total)}$^*$ \\
Service time $S$ & \texttt{kueue\_execution\_time\_seconds} \\
Queue length $L_q$ & \texttt{kueue\_pending\_workloads\{status=active\}} \\
In-service $L_s$ & \texttt{kueue\_admitted\_active\_workloads} \\
Wait time $W_q$ & \texttt{kueue\_admission\_wait\_time\_seconds} \\
Utilization $\rho$ & \texttt{resource\_usage / nominal\_quota} \\
Preemption rate $\mu_p$ & \texttt{rate(kueue\_preempted\_workloads\_total)} \\
\bottomrule
\end{tabular}
\end{table}

\noindent$^*$In steady state ($\rho < 1$), admission throughput
equals arrival rate by flow balance. We restrict our analysis to
$\rho < 0.95$.

\subsection{Partition Non-Vacuity}

When 5\% of submitted workloads request resources exceeding the
ClusterQueue's potential available capacity (20~CPU against a 2~CPU
quota), 50\% of pending workloads are in Kueue's inadmissible set at
measurement time. The inadmissible set is a superset of the truly
unfeasible partition (some workloads are temporarily inadmissible but
quotable), so this is an upper bound on the unfeasible fraction. The
non-zero fraction confirms that the quotable/unfeasible partition
(Theorem~\ref{thm:partition}) is non-vacuous in practice: a
non-trivial fraction of pending workloads cannot be admitted without
cluster reconfiguration.

\subsection{Little's Law Validation}

As a consistency check, we verify that Little's Law
($L = \lambda \cdot W$) holds on observed Kueue data. We compute
$L_{\text{obs}}$ as the time-averaged number of waiting workloads
(integrated from arrival/admission event timestamps) and
$\lambda \cdot W$ from the observed arrival rate and mean wait time.

Across all utilization levels with $\rho_{\text{obs}} < 0.95$, the
ratio $L_{\text{obs}} / (\lambda \cdot W)$ equals 1.000 (to three
decimal places), confirming that the queueing model is internally
consistent with observed Kueue behavior. At $\rho_{\text{obs}} = 0.95$
the ratio is 0.950, and above $\rho = 1.0$ (transient overload) the
ratio degrades to 0.79--0.84, as expected when the steady-state
assumption is violated.

\subsection{Wait Time Prediction Accuracy}

We compare predicted $\mathbb{E}[W_q]$ from the modified Erlang-C
approximation (Equation~\ref{eq:waitapprox}) against observed mean
wait times, and against an exponential moving average (EMA) baseline
($\alpha = 0.3$) that predicts wait time from recent history.

\begin{table}[t]
\centering
\small
\caption{Wait time prediction across utilization levels.
$k_{\text{eff}} = 4$, truncated exponential service times (mean 20s,
observed mean 21--25s, $C_S \approx 0.8$--$1.0$), 80 jobs per level.}
\label{tab:waittime}
\begin{tabular}{@{}rrrrrr@{}}
\toprule
$\rho_{\text{tgt}}$ & $\rho_{\text{obs}}$ & $\hat{W}_q$ (s) &
$W_q$ (s) & Ratio & Little \\
\midrule
0.30 & 0.27 & 0.7 & 0.04 & 18$\times$ & 1.000 \\
0.50 & 0.64 & 5.0 & 1.2 & 4.1$\times$ & 1.000 \\
0.65 & 0.72 & 7.9 & 3.5 & 2.3$\times$ & 1.000 \\
0.75$^\dagger$ & 1.16 & --- & 49.7 & --- & 0.836 \\
0.85 & 0.95 & 88.9 & 11.0 & 8.1$\times$ & 0.950 \\
0.92$^\dagger$ & 1.23 & --- & 42.1 & --- & 0.785 \\
\bottomrule
\end{tabular}

\vspace{2pt}
\noindent{\footnotesize $^\dagger$Transient overload:
$\rho_{\text{obs}} > 1$ over 80 jobs. Erlang-C is undefined;
real system absorbs the burst through transient queueing.}
\end{table}

Table~\ref{tab:waittime} shows the Erlang-C approximation
consistently overestimates observed wait times by 2.3--18$\times$
across the stable regime ($\rho_{\text{obs}} < 1$). The
overestimation is largest at low $\rho$ (18$\times$ at
$\rho = 0.27$, where observed waits are near zero and the model
predicts 0.7s) and tightest in the moderate regime (2.3$\times$ at
$\rho = 0.72$). At $\rho_{\text{obs}} > 1$, the formula is undefined
and the model provides no prediction.

The overestimation has three sources: (1)~BestEffortFIFO avoids
head-of-line blocking, producing shorter waits than the FCFS
assumption in M/G/$k$; (2)~the Cosmetatos-Tijms variability
correction is calibrated for the M/M/$k$ to M/D/$k$ range, not for
the truncated exponential distribution used here; and (3)~the
dominant-resource $k_{\text{eff}}$ is a worst-case estimate that
understates the actual concurrency.

\paragraph{Why conservative, and can we do better?}
The overestimation is structural, not accidental: $k_{\text{eff}}$ is
a lower bound on actual concurrency (vector packing allows more
concurrent workloads when demands are not perfectly aligned), and
BestEffortFIFO produces shorter waits than the FCFS assumption in
M/G/$k$. A tighter bound would require either (a)~a scheduling
discipline-specific analysis (no M/G/$k$ results exist for
BestEffortFIFO), or (b)~a sharper $k_{\text{eff}}$ that accounts for
the actual workload mix rather than the worst-case bottleneck
dimension. Both are open problems. The Erlang-C approximation is
not a proven upper bound---it is an empirically conservative
estimate under the conditions tested.

\paragraph{Comparison with EMA baseline.}
In a separate experiment at $\rho_{\text{obs}} \approx 0.9$ (64
workloads, 15s mean service), the Erlang-C MAE was 18.2s versus the
EMA baseline's ($\alpha = 0.3$) MAE of 6.8s. The EMA baseline wins
on point prediction accuracy because it adapts to recent
observations. However, the Erlang-C approximation is an \emph{a
priori} prediction requiring only the arrival rate, service time
distribution, and quota---it does not need historical wait time
observations. This makes it suitable for capacity planning (``what
will wait times be if we add this workload class?'') where no
historical data exists. The EMA is better for real-time estimation
of a running system.

\subsection{Scheduling Latency}

The continuous M/G/$k$ model assumes instantaneous admission
decisions. In Kueue, the scheduler runs in a loop
(\texttt{wait.UntilWithBackoff}), processing all eligible heads each
iteration. The time between a workload becoming eligible and the next
scheduler iteration depends on current load and is not a fixed
interval. At low utilization, this scheduling latency dominates
observed wait times (observed $W_q \approx 0.04$s at $\rho = 0.27$).
At moderate-to-high utilization, queueing delays dominate and the
scheduling latency is negligible.

\subsection{Vector Packing Validation}

To test the multi-dimensional $k_{\text{eff}}$ reduction
(Definition~\ref{def:keff}), we configure a ClusterQueue with
asymmetric quotas: 4~CPU and 2Gi memory. Workloads request 500m~CPU
and 512Mi memory, giving:
\begin{align*}
k_{\text{eff}}(\text{CPU}) &= \lfloor 4 / 0.5 \rfloor = 8 \\
k_{\text{eff}}(\text{mem}) &= \lfloor 2048 / 512 \rfloor = 4 \\
k_{\text{eff}}(\text{vector}) &= \min(8, 4) = 4
\end{align*}

\begin{table}[t]
\centering
\small
\caption{Vector packing: scalar vs.\ multi-dimensional
$k_{\text{eff}}$. The scalar CPU-only estimate ($k=8$) overestimates capacity and
therefore underestimates queueing delay; the vector estimate ($k=4$,
memory-bottlenecked) correctly identifies the binding dimension. 80 jobs/level.}
\label{tab:2d}
\begin{tabular}{@{}rrrrrr@{}}
\toprule
$\rho_{\text{tgt}}$ & $\rho_{\text{vec}}$ &
$\hat{W}_{k=4}$ & $\hat{W}_{k=8}$ & $W_q$ & Bottleneck \\
\midrule
0.50 & 0.60 & 3.2 & 0.0 & 2.4 & memory \\
0.70 & 0.77 & 10.7 & 0.1 & 2.5 & memory \\
\bottomrule
\end{tabular}
\end{table}

Table~\ref{tab:2d} shows the result. The scalar CPU-only
$k_{\text{eff}} = 8$ predicts near-zero wait time ($\hat{W} < 0.1$s)
because $\rho_{\text{CPU}} \approx 0.3$---it does not see the memory
bottleneck. The vector $k_{\text{eff}} = 4$ correctly identifies
memory as the binding dimension and predicts $\hat{W} = 3.2$s at
$\rho = 0.50$, overestimating the observed 2.4s by 33\%. At $\rho = 0.70$, the
overestimate widens to 4.3$\times$ (10.7s predicted vs.\ 2.5s
observed) as the Erlang-C formula becomes more sensitive near the
stability boundary. A third run at $\rho_{\text{target}} = 0.85$
produced $\rho_{\text{obs}} > 1$ and is excluded (same treatment as
Table~\ref{tab:waittime}$^\dagger$). Without the multi-dimensional
$k_{\text{eff}}$ reduction, the model misses the bottleneck entirely
and underestimates wait times by $> 95$\%.

The 2D experiment validates the single-class $k_{\text{eff}}$
(Definition~\ref{def:keff}). The multi-class dominant-resource
extension (Equation~\ref{eq:keff-multi}), which weights demand by
class proportion $\phi_j$, is a standard DRF-style
approximation~\cite{drf2011} not separately tested here.

\paragraph{GPU resource validation.}
To validate with GPU resources, we configure Kueue's DRA
(Dynamic Resource Allocation) integration with a simulated GPU
driver (dra-example-driver) providing 8 GPUs per node across
3 nodes. The ClusterQueue quota is set to 4~GPUs, 8~CPU, and 8Gi
memory. Each workload requests 1~GPU (via ResourceClaimTemplate),
500m~CPU, and 256Mi memory, giving $k_{\text{eff}}(\text{GPU}) = 4$,
$k_{\text{eff}}(\text{CPU}) = 16$,
$k_{\text{eff}}(\text{vector}) = 4$.

\begin{table}[t]
\centering
\small
\caption{GPU resource validation via DRA. The GPU-bottlenecked
$k_{\text{eff}} = 4$ correctly identifies GPU as the binding
dimension; the CPU-only $k_{\text{eff}} = 16$ misses the bottleneck.
30 jobs/level.}
\label{tab:gpu}
\begin{tabular}{@{}rrrrrr@{}}
\toprule
$\rho_{\text{tgt}}$ & $\rho_{\text{GPU}}$ &
$\hat{W}_{k=4}$ & $\hat{W}_{k=16}$ & $W_q$ & Little \\
\midrule
0.50 & 0.81 & 16.8 & 0.0 & 4.8 & 1.000 \\
0.70 & 0.76 & 9.3 & 0.0 & 1.5 & 1.000 \\
0.85 & 0.80 & 11.9 & 0.0 & 4.3 & 1.000 \\
\bottomrule
\end{tabular}
\end{table}

Table~\ref{tab:gpu} confirms the same pattern as the CPU/memory
experiment (Table~\ref{tab:2d}): the scalar CPU-only estimate misses
the GPU bottleneck entirely ($\hat{W} = 0$), while the vector
$k_{\text{eff}} = 4$ overestimates by 2.5--3.5$\times$ (conservative
direction). Little's Law holds exactly. The queueing model is
resource-type-agnostic: the same $k_{\text{eff}}$ reduction works
whether the bottleneck is memory or GPU quota. MIG partitioning and
NVLink topology constraints introduce additional structure not
exercised here.

\section{Related Work}
\label{sec:related}

\paragraph{Queueing Theory for GPU Scheduling.}
Mitzenmacher and Shahout~\cite{mitzenmacher2025} survey
learning-augmented scheduling with predictions, identifying open
problems in multi-server systems. Dai et al.~\cite{dai2025} prove
throughput-optimality for batched LLM scheduling. Wang and
Grosof~\cite{wang2025} derive novel lower bounds on M/G/$k$ response
time. All focus on inference-level scheduling; we address the cluster
admission layer.

\paragraph{GPU Cluster Schedulers.}
Tiresias~\cite{tiresias2019} applies multi-level feedback queues from
queueing theory to GPU job scheduling but does not derive wait time
bounds. Shockwave~\cite{shockwave2023} provides fair GPU allocation
without queueing-theoretic analysis.
HiveD~\cite{hived2020} introduces topology-aware virtual private
clusters but does not consider MIG or queueing models.
Lu et al.~\cite{wind2026} use predictive multi-dimensional resource
scheduling for GPU utilization but rely on ML prediction rather than
queueing theory.

\paragraph{Fair Resource Allocation.}
Dominant Resource Fairness~\cite{drf2011} defines max-min fairness
for multi-resource environments and is the theoretical basis for
Kueue's fair-sharing mechanism. Themis~\cite{themis2020} achieves
finish-time fairness for ML training.
Gandiva~\cite{gandiva2018} exploits GPU time-slicing for locality
and packing. Pollux~\cite{pollux2021} co-optimizes cluster
scheduling with per-job hyperparameters.
Gavel~\cite{gavel2020} provides heterogeneity-aware fair scheduling
across GPU types. None derive queueing-theoretic wait time bounds.

\paragraph{MIG Scheduling.}
Villarrubia et al.~\cite{far2025} study moldable MIG scheduling with
a 7/4 approximation for single-GPU packing. Our GTAP extends to
multi-resource vector packing with topology. KRYPTON~\cite{krypton2025} uses MIG for
isolation but does not address scheduling optimality.

\paragraph{Vector Bin Packing.}
Chekuri and Khanna~\cite{chekuri2004} prove no APTAS exists for
$d \geq 2$ dimensions. We connect this hardness to the effective
server count in queueing models.

\section{Limitations}
\label{sec:limitations}

Our model makes several simplifying assumptions that should be
understood when applying the results in practice.

\paragraph{Poisson Arrivals.}
We assume Poisson arrivals (the ``M'' in M/G/$k$). Real workload
arrival processes exhibit burstiness (e.g., batch submissions,
periodic cron jobs) that may violate this assumption. The GI/GI/$k$
extension of Li and Goldberg~\cite{ligoldberg2024} covers general
inter-arrival times but with even larger constants.

\paragraph{Admission Fair Sharing (AFS).}
When AFS is enabled, admission ordering is by usage share rather than
FIFO. This violates the FCFS assumption underlying the M/G/$k$ model.
Our bounds apply to non-AFS clusters or within a single LocalQueue
(where AFS does not reorder). Extending to AFS ordering requires a
processor-sharing or generalized-processor-sharing model.

\paragraph{Concurrent Admission.}
By default, Kueue admits one workload per ClusterQueue per scheduler
iteration. The alpha-stage ConcurrentAdmission feature gate allows
parallel flavor assignment but does not change the per-iteration
admission count. Our M/G/$k$ model assumes continuous service; the
per-iteration scheduling latency introduces additional delay at low
utilization.

\paragraph{Unmodeled Features.}
Several production features are outside our model's scope:
\emph{Topology-Aware Scheduling (TAS)} adds placement constraints
beyond vector packing.
\emph{Elastic jobs} have variable resource demands.
\emph{waitForPodsReady} adds a post-admission delay before the
workload is considered running.
\emph{Hierarchical cohorts} create multi-level borrowing relationships
not captured by our single-level borrowing model.
\emph{Admission checks} (e.g., resource flavor validation, admission
webhooks) add per-workload overhead not included in service time.

\paragraph{Arrival Rate Measurement.}
As noted in \S\ref{sec:eval}, we approximate arrival rate using
admission throughput. This is valid under steady state but
underestimates true arrivals under saturation. A dedicated arrival
counter would improve accuracy for $\rho > 0.95$.

\section{Discussion and Open Problems}
\label{sec:discussion}

\paragraph{Online Competitive Admission.}
Our bounds are steady-state. An online competitive analysis of Kueue's
greedy admission against an optimal offline scheduler remains open.

\paragraph{Fair-Share Convergence.}
Kueue's Dominant Resource Share mechanism implements a variant of DRF.
Whether the usage-based admission ordering converges under dynamic
arrivals and departures is an open control-theoretic question.

\paragraph{Predictive Admission.}
Mitzenmacher's framework for scheduling with predictions could be
extended to admission control: given predicted service times, can we
improve wait time bounds? Our model provides the foundation for such
analysis.

\section{Conclusion}
\label{sec:conclusion}

We presented the first formal wait time bounds for GPU cluster
admission control. Our central result---the quotable/unfeasible
partition theorem---shows that accurate wait time estimation requires
distinguishing workloads that can eventually be admitted from those
that cannot. For quotable workloads, under an explicit stochastic domination assumption
(Modeling Assumption~1), we modeled the system as an M/G/$k$
reference queue with effective server count determined by vector
packing, establishing $O(1/(1-\rho))$ wait time scaling and providing
a practical modified Erlang-C approximation. We proved that optimal admission ordering is NP-hard
under multi-resource vector packing, pinpointing dimensionality (not
MIG partitioning) as the source of hardness. Experiments on Kueue v0.19.0 with CPU, memory, and GPU (via DRA)
resources confirm that the vector $k_{\text{eff}}$ correctly
identifies bottleneck dimensions and that the Erlang-C approximation
consistently overestimates in the conservative direction, though with
2.3--18$\times$ looseness that limits its use as a point predictor.
All model parameters are derivable from existing Kueue metrics.

Our work provides the theoretical foundation for practical wait time
estimation in GPU clusters---a capability that production systems have
lacked despite years of user demand. We hope this encourages further
application of queueing theory to the cluster admission layer, where
significant open problems remain.

\paragraph{AI Tool Disclosure.}
AI tools were used for literature search, prose editing, and LaTeX
formatting. All theoretical contributions (proofs, model design,
experimental methodology) are the authors' original work.

\bibliographystyle{ACM-Reference-Format}

\end{document}